\documentclass[letterpaper, 10 pt, conference]{ieeeconf}  

\IEEEoverridecommandlockouts                              

\usepackage{cite}
\usepackage{amsmath,amssymb,amsfonts}
\usepackage{algorithmic}
\usepackage{graphicx}
\usepackage{textcomp}
\usepackage{booktabs}
\usepackage{subcaption}
\usepackage{xcolor}

\usepackage{bm}

\newtheorem{proposition}{Proposition}

\newtheorem{lemma}{Lemma}



\makeatletter
\renewenvironment{proposition}[1][]{%
  \refstepcounter{proposition}%
  \trivlist
  \item[\hskip\labelsep\bfseries Proposition~\theproposition%
    \if\relax\detokenize{#1}\relax.\else\ (#1).\fi]%
  \itshape
}{%
  \endtrivlist
}
\makeatother

\def\BibTeX{{\rm B\kern-.05em{\sc i\kern-.025em b}\kern-.08em
    T\kern-.1667em\lower.7ex\hbox{E}\kern-.125emX}}
\title{\LARGE \bf
Stabilizing Low-Stiffness Soft Robots at Unstable Equilibria: A Model-Based Optimal Control Approach
}

\author{Nikhil Nair$^{1}$, Daniele Caradonna$^{2}$, Anup Teejo Mathew$^{3}$, Federico Renda$^{3}$, \\ 
Cosimo Della Santina$^{1}$ and Daniel Feliu-Talegon$^{1}$ 
\thanks{*This work was supported in part by the Dutch Research Foundation (NWO) through the VENI grant ROSES 20297, as part of Project EMERGE and in part by the European Union (ERC, RIPLEY, 101165078).}
\thanks{$^{1}$Nikhil Nair, Daniel Feliu-Talegon  and Cosimo Della Santina are with the Department of Cognitive Robotics,
Delft University of Technology, 2628 CD Delft, The Netherlands, (e-mail: nikhilnicky972@gmail.com; d.feliutalegon@tudelft.nl; C.DellaSantina@tudelft.nl.}%
\thanks{$^{2}$Daniele Caradonna is 
 with Department of Excellence in Robotics and AI, Scuola Superiore Sant’Anna, Pisa, Italy. (e-mail: daniele.caradonna@santannapisa.it).}%
 \thanks{$^{3}$Anup Teejo Mathew and Federico Renda are with Department of Mechanical and Nuclear Engineering, Khalifa University of Science and Technology, Abu Dhabi, UAE (e-mail: anup.mathew@ku.ac.ae; federico.renda@ku.ac.ae).}
}

\begin{document}

\maketitle
\thispagestyle{empty}
\pagestyle{empty}

\begin{abstract}
Continuum soft robots are highly nonlinear and underactuated systems, which makes stabilizing them around unstable equilibrium points difficult. This challenge becomes even greater in low-stiffness regimes, where gravitational forces dominate elastic forces. In this letter, we explore whether optimal control can address these challenges using a computationally efficient method suitable for real-time use. We propose a nonlinear model predictive control (NMPC) approach based on Differential Dynamic Programming (DDP), combined with an accurate reduced-order Variable Strain model with analytical derivatives. This results in an efficient control framework for addressing the stabilization problem. Tested on a soft inverted pendulum and other hybrid soft–rigid systems, the method outperforms traditional controllers, stabilizes more challenging configurations, and operates reliably at real-time speeds (50-150 Hz), demonstrating both effectiveness and practical feasibility in experimental settings.
\end{abstract}

\section{Introduction}
\label{sec:introduction}

Soft robotics has grown rapidly over the past decade, driven by advances in materials, fabrication, and new technologies \cite{laschi2016soft}. Although soft robotics have advanced substantially, their widespread deployment remains limited by the difficulty of achieving precise and reliable control \cite{della2023model}. 


This perspective has recently shifted, primarily due to the development of reduced-order models that provide tractable and interpretable representations of soft robot dynamics \cite{armanini2023soft}. Beyond the well-known Piecewise Constant Curvature (PCC) assumption for modeling, more general reduced-order formulations—including Geometrically Variable Strain (GVS) approaches \cite{ mathew2025reduced} and polynomial curvature models \cite{della2019control}—offer dynamically consistent, finite-dimensional yet accurate Cosserat-rod-based representations that are well suited for control design.

However, underactuation remains a central challenge in the control of soft robots, even when reduced-order formulations are employed, as these models typically exhibit significantly more degrees of freedom than actuation inputs. Consequently, most model-based controllers provide formal guarantees only for restricted subclasses of systems, particularly those amenable to collocated feedback \cite{della2020soft} or to collocated Partial Feedback Linearization (PFL) \cite{spong1994partial}, which enables trajectory tracking of the actuated coordinates but only above certain stiffness thresholds. Linearization-based Linear Quadratic Regulator (LQR) control extends stabilization beyond purely collocated control. Nonlinear strategies—such as energy-injection switching control \cite{trumic2022stability} and closed-form passivity-based control in Hamiltonian form \cite{pagnanelli2024assessing}—have demonstrated effective solutions but are likewise constrained by stiffness limitations.

Optimal control is widely used to stabilize underactuated rigid systems around their unstable equilibrium configurations; however, its real-time application to soft robots remains challenging. Model Predictive Control (MPC), valued for its explicit handling of constraints and uncertainties \cite{norouzi2021constrained, hachen2025non}, is commonly employed with learned models \cite{haggerty2023control}. Data-driven methods \cite{yang2024model, bruder2020data} are typically integrated within MPC frameworks, while reduced-order FEM approaches enable real-time receding-horizon control \cite{tonkens2021soft, menager2024condensed} and offline trajectory optimization \cite{sanders2023dynamically}. Fully nonlinear MPC remains limited to simplified or learned models. A notable exception is \cite{caradonna2026model}, which investigates analytical model-based optimal control for soft robots. However, its integration with nonlinear optimal control has so far only been demonstrated in simulation, without real-time validation or evidence of improved stabilization around unstable equilibria in low-stiffness regimes.

Despite recent progress, important gaps remain: collocated methods are constrained by stiffness limitations, non-collocated strategies do not fully overcome low-stiffness challenges, and optimal control implementations rely either on simplified models with limited generalization capability or on high-dimensional models that are impractical for real-time deployment. Then, the central research question we address is: \textit{How can real-time optimal control methods be used to accurately and efficiently control soft robots operating in low-stiffness regimes around unstable equilibrium configurations?}

We develop a control approach that combines an analytical reduced-order model of soft robot configurations with a nonlinear model predictive control (NMPC) framework. By exploiting analytical dynamics gradients, the method enables efficient Differential Dynamic Programming (DDP) iterations for real-time closed-loop control. The approach stabilizes systems in mass–stiffness regimes where zero-dynamics constraints prevent both collocated and non-collocated model-based controllers from stabilizing unstable equilibria, while maintaining robustness and generalization across soft and hybrid robotic systems within a unified modeling–control framework.


\section{Preliminaries}

\subsection{Lagrangian form of soft robots}

A slender soft body of length \( L \) can be modeled as a Cosserat rod, represented as a continuous stack of rigid cross-sections parameterized by a curvilinear abscissa \( s \in [0,\, 1] \). Each cross-section is assigned a rigid-body configuration through the homogeneous transformation $\boldsymbol{g}(s) \in SE(3)$, which defines the directed spatial curve $\boldsymbol{g}(\bullet):s\rightarrow \boldsymbol{g}(s)$. The rod configuration is parameterized by a strain field, which is discretized using a finite set of generalized coordinates \cite{mathew2025reduced}. The rod dynamics in these coordinates are obtained by integrating the free-body equations derived from Cosserat equilibrium along the rod length, projecting them onto the generalized-coordinate space, $\bm{q}$, via the Jacobian, and incorporating the elastic, damping, and actuation laws. The resulting dynamics can be written in the form:
\begin{equation} \label{eq:GVSdynamics}
    \bm{M}(\bm{q}) \ddot{\bm{q}} + \bm{C}(\bm{q},\dot{\bm{q}}) \dot{\bm{q}}  + \bm{K}\bm{q} + \bm{G}(\bm{q}) + \bm{D}\dot{\bm{q}}  = \bm{A}(\bm{q}) \, \bm{u},
\end{equation}
\noindent where $\bm{M}(\bm{q}) \in \mathbb{R}^{n \times n}$, is the generalized mass matrix, $ \bm{C}(\bm{q}, \dot{\bm{q}}) \in \mathbb{R}^{n \times n}$ is the Coriolis matrix, $\bm{D} \in \mathbb{R}^{n \times n}$ is the elastic damping matrix, $ \bm{K} \in \mathbb{R}^{n \times n}$ is the stiffness matrix, $\bm{G}(\bm{q}) \in \mathbb{R}^{n}$ is the gravitational
force. The term $\bm{A}(\bm{q}) \in \mathbb{R}^{n \times n_{\mathrm{a}}}$ denotes the actuation matrix, which maps control inputs $\bm{u}$ into the Lagrangian generalized forces $\bm{\tau}_{\mathrm{a}}$.

\subsection{Collocated form}

A Lagrangian system of the form (\ref{eq:GVSdynamics}) can be conveniently reformulated by separating the dynamic equations of the actuated $\bm{\theta_\mathrm{a}} \in \mathbb{R}^{n_\mathrm{a}}$ and unactuated variables  $\bm{\theta_\mathrm{u}} \in \mathbb{R}^{n - n_\mathrm{a}}$, such that $\bm{\theta} = [\bm{\theta_\mathrm{a}}^{T}, \, \bm{\theta_\mathrm{u}}^{T}]^{T}$. Finding a more suitable representation of the system dynamics can greatly simplify both analysis and control design through an appropriate transformation of the generalized coordinates. This procedure is described in detail in \cite{pustina2024input}. After performing the change of coordinates, the decoupled dynamics of the system can be expressed as:
\begin{equation}\label{eq:GVSdynamics_2}
\begin{aligned}
\begin{bmatrix}
\tilde{\bm{M}}_\mathrm{aa} & \tilde{\bm{M}}_\mathrm{au} \\
\tilde{\bm{M}}_\mathrm{ua} & \tilde{\bm{M}}_\mathrm{uu}
\end{bmatrix}
\begin{bmatrix}
\ddot{\bm{\theta}}_\mathrm{a} \\
\ddot{\bm{\theta}}_\mathrm{u}
\end{bmatrix}
&+
\begin{bmatrix}
\tilde{\bm{C}}_\mathrm{aa} & \tilde{\bm{C}}_\mathrm{au} \\
\tilde{\bm{C}}_\mathrm{ua} & \tilde{\bm{C}}_\mathrm{uu}
\end{bmatrix}
\begin{bmatrix}
\dot{\bm{\theta}}_\mathrm{a} \\
\dot{\bm{\theta}}_\mathrm{u}
\end{bmatrix}
+
\begin{bmatrix}
\tilde{\bm{G}}_\mathrm{a} \\
\tilde{\bm{G}}_\mathrm{u}
\end{bmatrix} \\
+
\begin{bmatrix}
\tilde{\bm{K}}_\mathrm{aa} & \tilde{\bm{K}}_\mathrm{au} \\
\tilde{\bm{K}}_\mathrm{ua} & \tilde{\bm{K}}_\mathrm{uu}
\end{bmatrix}
\begin{bmatrix}
\bm{\theta}_\mathrm{a} \\
\bm{\theta}_\mathrm{u}
\end{bmatrix}
&+
\begin{bmatrix}
\tilde{\bm{D}}_\mathrm{aa} & \tilde{\bm{D}}_\mathrm{au} \\
\tilde{\bm{D}}_\mathrm{ua} & \tilde{\bm{D}}_\mathrm{uu}
\end{bmatrix}
\begin{bmatrix}
\dot{\bm{\theta}}_\mathrm{a} \\
\dot{\bm{\theta}}_\mathrm{u}
\end{bmatrix}
=
\begin{bmatrix}
\bm{u} \\
\bm{0}
\end{bmatrix}.
\end{aligned}
\end{equation}

where, for compactness, we define, $\tilde{\bm{h}} = \tilde{\bm{C}}\dot{\bm{\theta}} + \tilde{\bm{G}} + \tilde{\bm{K}}\bm{\theta} + \tilde{\bm{D}}\dot{\bm{\theta}}$. Based on the general form in (\ref{eq:GVSdynamics_2}), different control strategies can be employed. Here, we focus on open-loop and closed-loop approaches. In particular, we consider the following purely feedforward controller and collocated PFL:

\begin{equation}\label{eq1}
\bm{u} = \tilde{\bm{G}}_a(\overline{\bm{\theta}})+\tilde{\bm{K}}_{aa}(\overline{\bm{\theta}})\overline{\bm{\theta}}_a+\tilde{\bm{K}}_{au}(\overline{\bm{\theta}})\overline{\bm{\theta}}_u\,
\end{equation}

\vspace{-0.5cm}

\begin{equation}\label{eq2}
    \bm{u} =
    \left( \tilde{\bm{h}}_a -
    \tilde{\bm{M}}_\mathrm{au}\tilde{\bm{M}}^{-1}_\mathrm{uu}
    \tilde{\bm{h}}_u \right)
    +
    \left(
    \tilde{\bm{M}}_\mathrm{aa}
    -\tilde{\bm{M}}_\mathrm{au}\tilde{\bm{M}}^{-1}_{uu}\tilde{\bm{M}}_\mathrm{ua}
    \right)
    \bm{\tau}, \,
\end{equation}

\vspace{-0,3cm}

\noindent where $ \bm{\tau} = -\bm{K}_P \tilde{\bm{\theta}}_\mathrm{a} - \bm{K}_D \dot{\bm{\theta}}_\mathrm{a};\bm{K}_P \succ \bm{0}; \bm{K}_D \succ \bm{0}\,$


\section{Methodology}

\subsection{Instability Analysis}

For underactuated systems such as soft robots, zero dynamics describe the internal behavior of the unactuated coordinates when the outputs are constrained by control. Their stability determines whether the internal dynamics remain bounded. We show here that specific combinations of stiffness and gravity terms relative to mass terms, the system cannot be stabilized using state of the art control approaches.

\begin{proposition}[Open-Loop instability]
Consider the open-loop soft robot dynamics (\ref{eq:GVSdynamics_2}) subject to a constant control action (\ref{eq1}) and suppose that $(\overline{\bm{\theta}},\bm 0)$ be an equilibrium. Assume the inertia and damping matrices satisfy $\tilde{\bm{M}} \succ \bm{0}$ and $\tilde{\bm{D}} \succ \bm{0}$. Let the Hessian of the potential energy at the equilibrium be:

\begin{equation}\label{rest1}
\bm{H} = \tilde{\bm{K}} + \left.\nabla_{\bm{\theta}} \tilde{\bm{G}}\right|_{\bm{\theta}=\overline{\bm{\theta}}} \,
\end{equation}

If $\bm H$ possesses at least one negative eigenvalue, i.e., $\bm H \not\succeq 0$,
then the equilibrium $(\overline{\bm{\theta}},\bm 0)$ is unstable.

\end{proposition}

\begin{proof}
Define the displacement variable $\bm x = \bm{\theta}-\overline{\bm{\theta}}$.  
Linearizing the dynamics at the equilibrium $(\overline{\bm{\theta}},\bm 0)$ yields $\tilde{\bm{M}}_e \ddot{\bm x} + \tilde{\bm{D}}_e \dot{\bm x} + \bm H \bm x = \bm{0}$, 
where $\tilde{\bm{M}}_e := \tilde{\bm{M}}(\overline{\bm{\theta}})$, $\tilde{\bm{D}}_e := \tilde{\bm{D}}(\overline{\bm{\theta}})$  and $\bm H = \nabla^2 U(\overline{\bm{\theta}})$, with $U(\bm{\theta})$ a potential function satisfying $\nabla U(\bm{\theta}) = \tilde{\bm{K}} \bm{\theta} + \tilde{\bm{G}}(\bm{\theta})$. Introducing the state vector $\bm z = [\bm x^\top\ \dot{\bm x}^\top]^\top$, the linearized system takes the form $\dot{\bm z} = \bm A \bm z$ with
\[
\bm A =
\begin{bmatrix}
\bm 0 & \bm I \\
-\tilde{\bm{M}}_e^{-1}\bm H & -\tilde{\bm{M}}_e^{-1}\tilde{\bm{D}}_e
\end{bmatrix}.
\]

The eigenvalues of $\bm A$ satisfy the characteristic equation $\det\!\left(\lambda^2 \tilde{\bm{M}}_e + \lambda \tilde{\bm{D}}_e + \bm H\right)=\bm{0}$.
Since $\tilde{\bm{M}}_e \succ 0$ and $\tilde{\bm{D}}_e \succ 0$ while $\bm H$ has negative eigenvalues ($\bm H \not\succeq 0$), the quadratic matrix polynomial above has at least one real positive eigenvalue. Consequently, the linearized system is unstable. By Lyapunov's indirect method, the equilibrium $(\overline{\bm{\theta}},\bm 0)$ of the nonlinear system is therefore unstable.
\end{proof}

%
%
%

\begin{proposition}[Collocated PFL instability]
Consider the soft robot dynamics in \eqref{eq:GVSdynamics_2} together with the control law (\ref{eq2}). Let the Hessian of the potential energy of the unactuated coordinates at the equilibrium be: 

\begin{equation}\label{rest1_2}
\bm{H}_u = \tilde{\bm{K}}_\mathrm{uu}
    +  \left.\nabla_{\bm{\theta}_u} \tilde{\bm{G}}_u  \right|_{\bm{\theta}=\bm{\theta}^{*}}
\end{equation}

If $\bm H_u$ possesses at least one negative eigenvalue, i.e., $\bm H_u \not\succeq 0$,
then the equilibrium $(\overline{\bm{\theta}},\bm 0)$ is unstable.

\end{proposition}

\begin{proof}
By introducing the control law (\ref{eq2}) in \eqref{eq:GVSdynamics_2}, it is straightforwardto show that the actuated variables $\bm{\theta}_a$ are regulated to a constant value $\bm{\theta}_a^{*}$. However, the presence of unactuated coordinates implies that the system still exhibits internal dynamics of dimension $2(n-n_a)$. These dynamics can be obtained by examining the residual dynamics in (\ref{eq:GVSdynamics_2}) after imposing $\dot{\bm{\theta}}_a=\bm{0}$ and $\ddot{\bm{\theta}}_a=\bm{0}$:
\begin{equation}\label{zero_dynamics_2}
\begin{split}
 \tilde{\bm{M}}_\mathrm{uu}(\bm{\theta}_a^{*},\bm{\theta}_u) & \ddot{\bm{\theta}}_u  +  \tilde{\bm{C}}_\mathrm{uu}(\bm{\theta}_a^{*},\bm{\theta}_u,\bm{0},\dot{\bm{\theta}}_u) \dot{\bm{\theta}}_u + \tilde{\bm{G}}_u(\bm{\theta}_a^{*},\bm{\theta}_u)    \\    & +  \tilde{\bm{K}}_\mathrm{ua}  \bm{\theta}_a^{*}  +  \tilde{\bm{K}}_\mathrm{uu}  \bm{\theta}_u + \tilde{\bm{D}}_\mathrm{uu}(\bm{\theta}_a^{*},\bm{\theta}_u) \dot{\bm{\theta}}_u  = \bm{0} \, . 
\end{split}
\end{equation}

\vspace{-0.4cm}

Linearizing the dynamics (\ref{zero_dynamics_2}) at the equilibrium $(\overline{\bm{\theta}},\bm 0)$ and proceeding similarly to Proposition 1. We can obtain that the eigenvalues of the unactuated submatrix $\bm A_{u}$ satisfy  the characteristic equation $\det\!\left(\lambda^2 \tilde{\bm{M}}_{uu,e} + \lambda \tilde{\bm{D}}_{uu,e} + \bm H_u\right)=\bm{0}$ where $\tilde{\bm{M}}_{uu,e} := \tilde{\bm{M}}_{uu}(\overline{\bm{\theta}})$, $\tilde{\bm{D}}_{uu,e} := \tilde{\bm{D}}_{uu}(\overline{\bm{\theta}})$.
Since $\tilde{\bm{M}}_{uu,e} \succ 0$ and $\tilde{\bm{D}}_{uu,e} \succ 0$ while $\bm H_u$ has negative eigenvalues ($\bm H_u \not\succeq 0$), the quadratic matrix polynomial above has at least one real positive eigenvalue. Consequently, the linearized system is unstable. By Lyapunov's indirect method, the equilibrium $(\overline{\bm{\theta}},\bm 0)$ of the nonlinear system is therefore unstable.

\end{proof}



\subsection{Remarks on Instability}

The stability of underactuated soft robotic systems is strongly influenced by the Hessian of the potential energy and the internal dynamics. In open loop, the equilibrium becomes unstable when the Hessian of the potential energy satisfies $\bm H \not\succeq \bm{0}$, which typically occurs when the system stiffness is insufficient for elastic forces to dominate gravity.

With collocated partial feedback linearization (PFL), closed-loop stability depends on the zero dynamics and becomes unstable when $\bm H_u \not\succeq \bm{0}$ at the equilibrium. Notably, $\bm H_u \not\succeq \bm{0}$ implies $\bm H \not\succeq \bm{0}$, whereas the converse does not necessarily hold, allowing collocated PFL to extend the range of low-stiffness regimes with respect to the open-loop case. Although non-collocated PFL linearizes unactuated coordinates and allows pole placement via linear feedback (e.g., linear–quadratic regulator (LQR)), no general stabilization method exists for soft robotic systems. Moreover, such controllers typically provide only local stabilization near the equilibrium and are not suitable for regulating trajectories starting far from the desired configuration. The region of attraction becomes very small in very low-stiffness regimes.

\subsection{Problem Formulation}

The objective of this work is to stabilize the robot around unstable equilibria, particularly in low-stiffness regimes where traditional collocated control strategies fail. Consider a slender soft robot modeled using an accurate reduced-order formulation (e.g., Polynomial Curvature or GVS-based models), resulting in a discretized nonlinear dynamic system of the form $\bm{x}_{k+1} = \bm{f}_d(\bm{x}_k, \bm{u}_k)$, where $\bm{x}_k \in \mathbb{R}^{2n}$ is the state vector (positions and velocities), $\bm{u}_k \in \mathbb{R}^{n_a}$ is the control input, and $\bm{f}_d$ represents the discretized system dynamics. The system is underactuated, i.e., $n_a < n$. The optimal control problem is formulated in discrete time as
\begin{equation}\label{eq:continous_ocp}
    \begin{aligned}
        \min_{\bm{u}_0, \ldots, \bm{u}_{N-1}} 
\; & \ell_f(\bm{x}_N)
+ \sum_{k=0}^{N-1} \ell(\bm{x}_k, \bm{u}_k) \\
        \textnormal{s.t.} \ \ \ \ \bm{x}_{k+1} &= \bm{f}_d(\bm{x}_k, \bm{u}_k), \quad k = 0, \ldots, N-1, \\
        \bm{x}_0 &= \bm{x}_{\mathrm{init}}, \ \bm{u}_k \in \mathcal{U}, \ \bm{x}_k \in \mathcal{X}.
    \end{aligned} 
\end{equation}
Here, $\ell(\bm{x}_k, \bm{u}_k)$ is the stage cost penalizing deviation from the target state and control effort, while $\ell_f(\bm{x}_N)$ is a terminal cost. The control objective is to regulate the robot to a desired unstable equilibrium $\bar{\bm{x}} =[\bar{\bm{q}}^{T}, \bm{0}^T]^T$.



\subsection{Differential Dynamic Programming for NMPC}\label{DDP_MPC}

DDP is an efficient trajectory optimization algorithm for nonlinear dynamical systems. For the shape regulation problem, the optimal control problem can be discretized with quadratic stage and terminal costs:
\begin{equation}
\small
\begin{aligned}
    \min_{\bm{u}_1,\cdots,\bm{u}_{N-1}} \ & \frac{1}{2}\tilde{\bm{x}}_N^T \bm{W}_N \tilde{\bm{x}}_N + \sum_{k=1}^{N-1} \left( \frac{1}{2}\tilde{\bm{x}}_k^T \bm{W}_k \tilde{\bm{x}}_k  + \frac{1}{2}\bm{u}_k^T \bm{R}_k \bm{u}_k  \right)  \\
    &\text{s.t.} \hspace{1cm} \bm{x}_{k+1} = \bm{f}_d(\bm{x}_k, \bm{u}_k), \quad \tilde{\bm{x}}_k  = (\bm{x}_k - \bm{x}_g) 
\end{aligned}
\end{equation}
\noindent where $\bm{W}_k$ and $\bm{R}_k$ are positive semi-definite and positive definite weighting matrices, respectively.

DDP is derived from Bellman's principle of optimality. Let $V_k^*(\bm{x})$ denote the optimal cost-to-go from time step $k$. The dynamic programming recursion is
\begin{equation}
V_k^*(\bm{x}) = \min_u \left( \ell(\bm{x},\bm{u}) + V_{k+1}^*(\bm{f}_d(\bm{x},\bm{u}) \right) \, ,
\end{equation}
\noindent where $\ell(\bm{x},\bm{u})$ denotes the stage cost. Since computing $V_k^*(\bm{x})$ exactly is intractable for nonlinear systems, DDP approximates the state--action value function
\begin{equation}
Q_k(\bm{x},\bm{u}) = \ell(\bm{x},\bm{u}) + V_{k+1}(\bm{f}_d(\bm{x},\bm{u})) \, ,
\end{equation}
with a second-order Taylor expansion around a nominal trajectory $(\bm{x}_k,\bm{u}_k)$. This local quadratic model allows the optimal control variation to be computed analytically as
\begin{equation}
\delta \bm{u}_k = -\bm{Q}_{uu,k}^{-1}(\bm{Q}_{u,k} + \bm{Q}_{ux,k}\delta \bm{x}_k) = -\bm{l}_k - \bm{K}_k \delta \bm{x}_k \, ,
\end{equation}
\noindent where $\bm{l}_k$ is the \textit{feedforward term} and $\bm{K}_k $ is the \textit{time-varying feedback term}. 
Notice that $\bm{Q}_{uu,k}, \bm{Q}_{u,k}, \bm{Q}_{ux,k}$ depends on the partial derivatives $\nabla_{\bm{x}}\bm{f}_d, \nabla_{\bm{u}}\bm{f}_d, \nabla^2_{\bm{x,x}} \bm{f}_d, \nabla^2_{\bm{u,u}} \bm{f}_d, \nabla^2_{\bm{x,u}}\bm{f}_d$.

During the \textit{backward pass}, these quantities are computed recursively together with the local quadratic approximation of the value function. A \textit{forward pass} then applies the updated control law and propagates the nonlinear dynamics to obtain a new trajectory. The process is repeated until the cost converges. The resulting feedforward and feedback terms provide a local control policy
\vspace{-0.2cm}

\begin{equation}\label{ec_rf}
\bm{u}_k = \bm{u}_k^\star - \bm{K}_k(\bm{x}_k - \bm{x}_k^\star) \, ,
\end{equation}

\noindent although in many implementations only the first optimized control input $\bm{u}_k^\star$ is applied to the system, the finite-horizon optimal control is shifted forward and solved again at the next step. 


\subsection{Implementation Details for Real-Time Deployment}

The successful real-time control of a soft robot systems in low-stiffness regimes requires both an accurate dynamic model and a control strategy that adapts to the robot motion in real time. To efficiently compute the system dynamics and their gradients required by DDP, we adopt the variable strain parametrization based on the GVS model \cite{mathew2025reduced}. The kinematics are computed recursively using the Magnus expansion with a fourth-order Zanna quadrature approximation of the variable strain field, combined with Gauss quadrature for efficient integration. This results in a recursive algorithm for computing the Lagrangian dynamics of the system.

We employ an iLQR-type DDP formulation, where the dynamics are linearized and the cost is quadratically approximated using first-order system derivatives, while second-order derivatives of the dynamics are neglected. The Jacobians $\nabla_{\bm{x}}\bm{f}_d, \nabla_{\bm{u}}\bm{f}_d$ can be analytically differentiated using the adapted RNEA for GVS introduced in \cite{mathew2026analytical}. To implement NMPC, the DDP solver is embedded in a receding-horizon framework. At each sampling instant $t_k$, a finite-horizon optimal control problem is solved using the current measured state as the initial condition. Unlike offline trajectory optimization, NMPC must operate in real time. Therefore, instead of iterating DDP until convergence, only a small number of backward–forward passes (1-3) are executed at each time step. To accelerate convergence, the solver is \emph{warm-started} using the shifted solution from the previous NMPC iteration. This allows NMPC to retain the predictive and optimizing properties of DDP while remaining computationally feasible for real-time control. The control framework was first prototyped and validated in simulation in MATLAB, where the DDP/NMPC pipeline, cost shaping, and model verification were developed. For the computational times reported here in simulation, as well as for real-time execution on hardware, the controller and system dynamics were reimplemented in C++ using the Crocoddyl library \cite{mastalli2020crocoddyl} for shooting-based optimal control.

\section{Simulation Results}

This section evaluates the NMPC/DDP controller in simulation and compares it with state-of-the-art techniques for controlling underactuated soft robots. The algorithm is first tested on the Soft Inverted Pendulum (SIP) with length $L = 1~\mathrm{m}$ and radius $r = 0.03~\mathrm{m}$. The material properties are $\rho = 10^{3}~\mathrm{kg/m^{3}}$, $\upsilon = 10^{4}~\mathrm{Pa\cdot s}$, $\nu = 0.5$, and $E = 1\times10^{6}~\mathrm{Pa}$. To demonstrate the superiority of the proposed approach over PFL techniques, especially in low-stiffness regimes, the mass and Young’s modulus parameters are subsequently varied. The NMPC simulation uses a time step $dt = 0.01~\mathrm{s}$, a prediction horizon of $N = 100$ steps ($1.00~\mathrm{s}$), and a total simulation length of $T = 2000$ steps ($20.00~\mathrm{s}$).


\subsection{Baseline Performance on the Soft Inverted Pendulum}

Our proposed methodology is first tested on the SIP system, where the base is clamped and a torque is applied at the tip to control the motion. The cost weights are defined as $\bm{W} = 0.1,\mathbf{I}_4$ for the running state cost, $\bm{R} = \mathrm{diag}(1)$ for the control cost, and $\bm{W}_N = \mathrm{diag}(10^{2},10^{2},10^{1},10^{1})$ for the terminal state cost. The system is initialized at $\bm{x}_0 = [8.5879,-11.5744,0,0,]^{T}$ and the target state is $\bm{x}_g = \bm{0}$, corresponding to the unstable upright equilibrium that the controller aims to stabilize. The approach efficiently achieves swing up trajectory as shown in Fig. \ref{fig:sip_strobe}.a, with a settling time of \(1.70~\mathrm{s}\) (5\% band). Control remains moderate (\(\|u\|_\infty\approx 9.77\) and the final cost vanishes. 

\begin{figure}[h]
  \centering
  \includegraphics[width=0.9\linewidth]{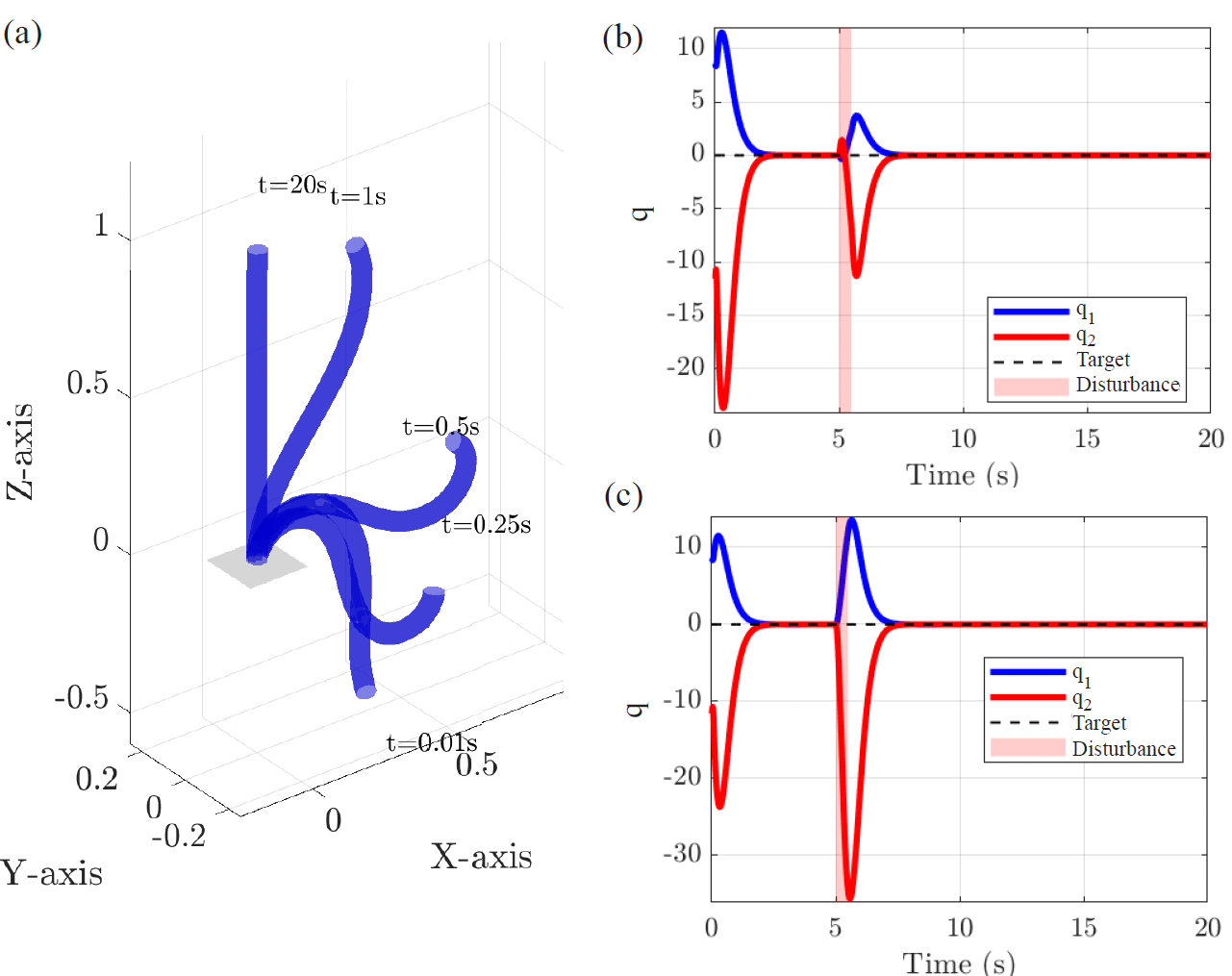}
  \caption{(a) Stroboscopic plot of the SIP during swing-up and stabilization under NMPC; (b) response to an external force of 3 N applied at the tip; and (c) response to an external force of 8 N applied at the midpoint.}
  \label{fig:sip_strobe}
\end{figure}


Closed-loop robustness was also evaluated by applying exogenous forces at the tip and midpoint of the link within the range $[0,15]\,\text{N}$ after the system completed the swing-up trajectory. The system demonstrates strong robustness to disturbances, with recovery times clustering around $1.1$--$1.4\,\text{s}$ for all cases, even for large forces applied at the tip. Fig. \ref{fig:sip_strobe}.b and Fig. \ref{fig:sip_strobe}.c show the results for a $3 \, \text{N}$ force applied at the tip and a $8 \, \text{N}$ force applied at the midpoint, respectively.

We also compared the use of analytical and numerical derivatives (see Fig.~\ref{fig:analytical_vs_numerical}). Employing analytical expressions provides a clear computational advantage, yielding approximately a $4.4\times$ reduction in the average per-step solve time (about $10.2\,\text{ms}$ vs.\ $45.2\,\text{ms}$). This improvement enables significantly higher control rates ($\approx 98\,\text{Hz}$ vs.\ $22\,\text{Hz}$) while maintaining comparable accuracy.

\begin{figure}[h]
  \centering
  \includegraphics[width=0.96\linewidth]{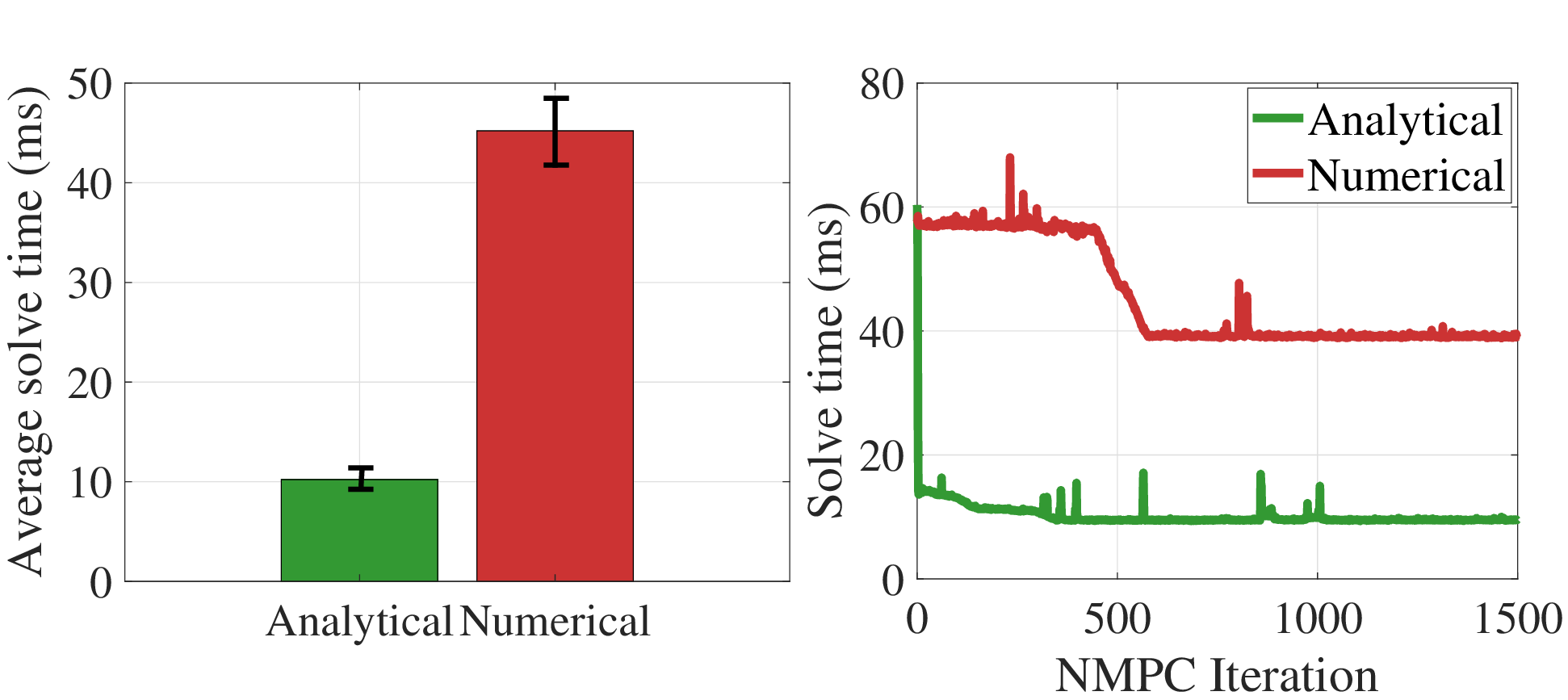}
  \caption{Analytical vs. numerical (finite-difference) dynamics gradients in NMPC/DDP.}
  \label{fig:analytical_vs_numerical}
\end{figure}

\subsection{Stabilization Capabilities of NMPC/DDP}

This section compares the NMPC/DDP controller with two controller baselines, \eqref{eq1} and \eqref{eq2}, for which the conditions under which the system starts to become unstable were previously derived. Apart from that, we also apply a non-collocated PFL controller developed in \cite{della2020soft}. This can be done for this system because the dimension of the unactuated part and the actuated are the same. Therefore, we can linearize the unactuated dynamics, obtaining $\ddot{\bm{\theta}}_u = \bm{u}$, and then apply LQR to the linearized system around the equilibrium. The objective is to demonstrate that NMPC can stabilize operating points in regions where the existing controllers fail. The evaluation is performed on the SIP system with the parameters mentioned previously, but including variations in Young's modulus and mass density.

Controller performance is evaluated over a grid of mass-stiffness pairs $m \in \{1,\dots,9\}, \ \sigma := \log_{10}(1/k) \in \{0,1,\dots,100\}$,
where the logarithmic parameterization emphasizes low-stiffness regimes. The stiffness parameter \(k = EJ_y \) collects the bending properties, while \(m=\rho A\) represents the inertia per unit length for this planar example.

Figure~\ref{fig:feasibility_regions} summarizes control feasibility across the \((m,k)\) plane. The bounds defined by blue curve and red curve delineate the region in which collocated PFL is theoretically valid. As $m$ increases, both thresholds shift downward, thereby widening the feasible band.

\begin{figure}[h]
  \centering
  \includegraphics[width=0.95\linewidth]{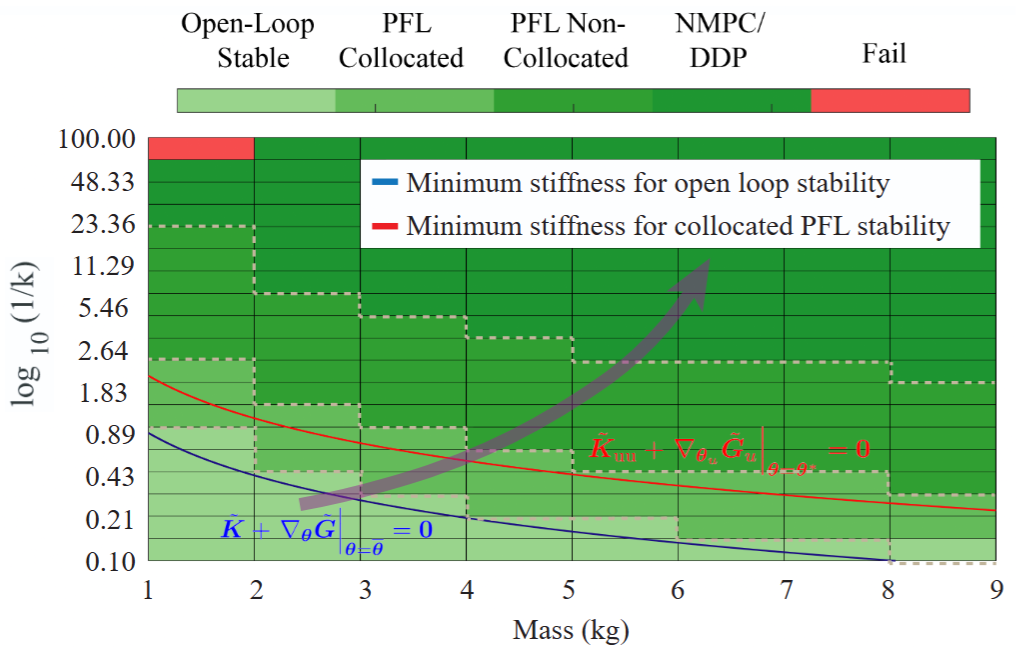} 
  \caption{Feasibility regions over \((m,k)\) for open-loop stability and collocated PFL, with analytically derived bounds. Non-collocated PFL and NMPC/DDP are evaluated beyond these analytical limits. Stable operation lies below each curve, and the purple arrow indicates decreasing stiffness.}
  \label{fig:feasibility_regions}
\end{figure}

Collocated PFL operates reliably only within a limited analytically defined region, while non-collocated PFL stabilizes a larger portion of the mass–stiffness plane at the cost of higher control effort and can handle significantly lower stiffness values. However, beyond a certain stiffness threshold, even the non-collocated approach fails despite the equilibrium remaining locally attractive. NMPC/DDP remains effective beyond the feasibility limits of both PFL methods, especially for high mass and low stiffness where unfavorable zero dynamics limit PFL, demonstrating the benefit of predictive optimal control.

\subsection{Scalability and Generalization}

The controller’s generalizability is tested on two additional soft robotic systems: the R-SIP and the Soft Furuta system. The R-SIP uses a rigid revolute joint at the base where the torque is applied, with a soft link whose bending is modeled using a second-order polynomial yielding to $n=3+1$. The Soft Furuta system consists of a rigid driving arm that rotates about the z-axis, with a soft appendage attached to the tip of the arm through a revolute joint. The deformation of the soft appendage is modeled using a linear polynomial in the curvature within the same framework yielding to $n=2+1$. Results in Fig.~\ref{fig:rsip_furuta_strobe} shows that NMPC/DDP successfully performs swing-up and stabilization for both systems, producing smooth transients, bounded control effort, and rapid convergence to the upright equilibrium despite underactuation and low stiffness.

\begin{figure}[h]
  \centering  \includegraphics[width=0.96\linewidth]{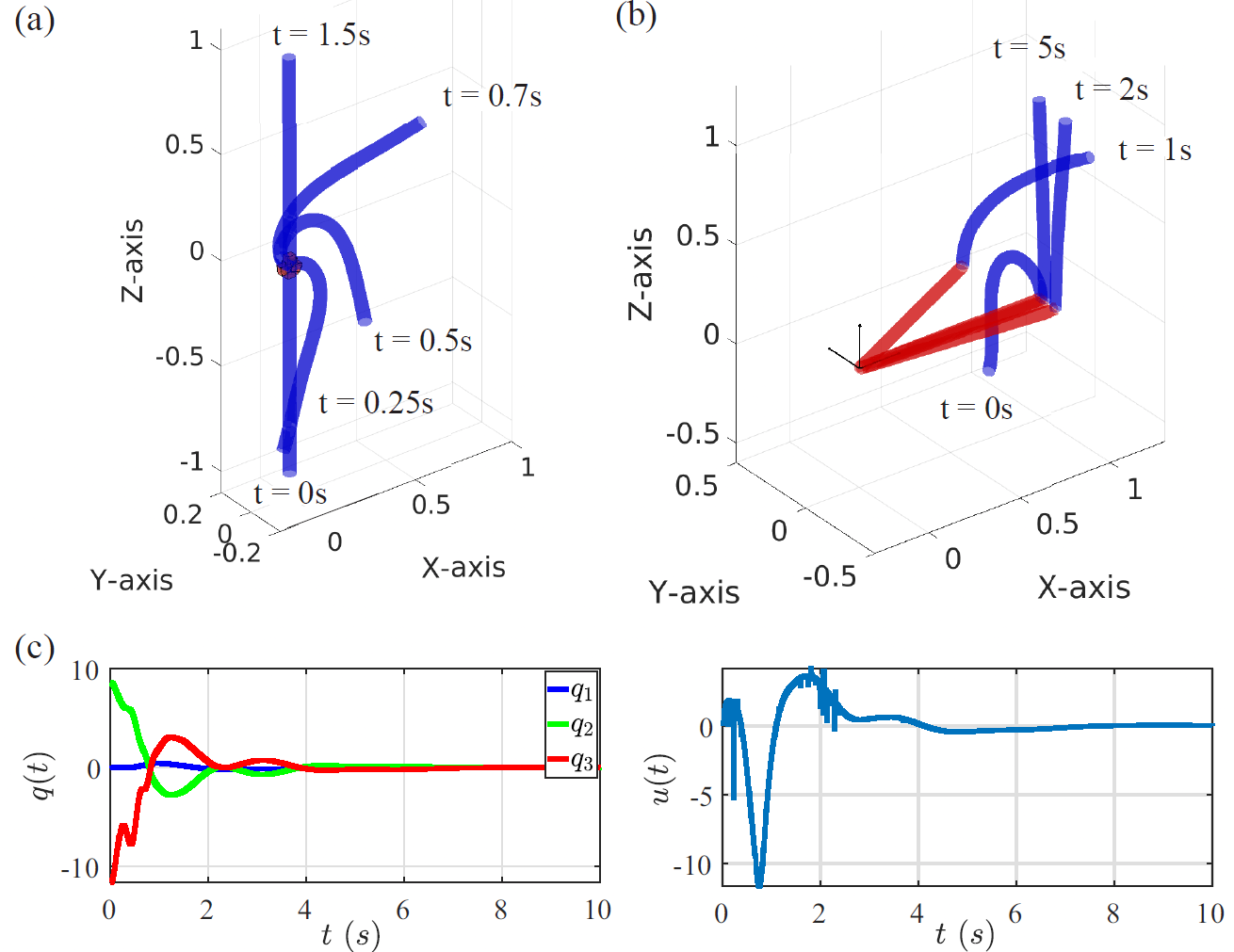}
  \caption{Simulation results: (a) shape evolution of the R-SIP; (b) shape evolution of the Soft Furuta system; and (c) state evolution and control input of the Soft Furuta system.}
  \label{fig:rsip_furuta_strobe}
\end{figure}


From a computational perspective, the approach scales well with model complexity. 
The per-step solve time increases from $7.05\,\text{ms}$ ($n=2$) to $27.61\,\text{ms}$ ($n=5$), with intermediate values of $17.19\,\text{ms}$ ($n=3$) and $25.53\,\text{ms}$ ($n=4$). 
This corresponds to achievable control frequencies ranging from approximately $150\,\text{Hz}$ to $50\,\text{Hz}$. 
These results indicate that the framework remains computationally tractable and maintains real-time performance as the number of deformation modes increases.


\section{Experimental Results}

The approach was experimentally validated on a $7$-DoF Franka Emika Panda manipulating a $0.76\,\mathrm{m}$ deformable cable grasped at a single point. The cable–gripper system is modeled as a R-SIP, composed of a rigid revolute joint at the gripper followed by a soft link representing the cable. The control task consists of swinging the system from the stable downward equilibrium to an unstable configuration and switching between two lateral equilibria. Material parameters were experimentally identified following~\cite{feliu2024dynamic}, yielding a Young’s modulus of $22.07\,\mathrm{MPa}$, density $3546.6\,\mathrm{kg/m^3}$, and elastic damping $2\,\mathrm{MPa\cdot s}$. Quadratic curvature polynomials in $y$ (planar movement) yield a $3$-DoF model of the soft system plus the degree of freedom of the revolute joint. The cable states $\mathbf{q}$ are estimated using an OptiTrack motion-capture system that tracks several markers. Their positions are then used to reconstruct the configuration through inverse kinematics. Fig.~\ref{fig:panda_swing_grid} presents snapshots of the experiments together with the desired shapes computed from the equilibrium model. Fig.~\ref{fig:nmpc_vs_optitrack_states} compares the experimental trajectories obtained with the NMPC controller to the corresponding simulated NMPC solution. The close agreement indicates that the robot successfully reproduces the planned motion. Minor deviations are mainly caused by modeling inaccuracies and cable hysteresis (left/right asymmetry). The reference is tracked sufficiently well to complete the task despite the Franka robot’s joint-speed limits not being explicitly enforced as hard NMPC constraints.


\begin{figure}[h]
  \centering
\includegraphics[width=0.95\linewidth]{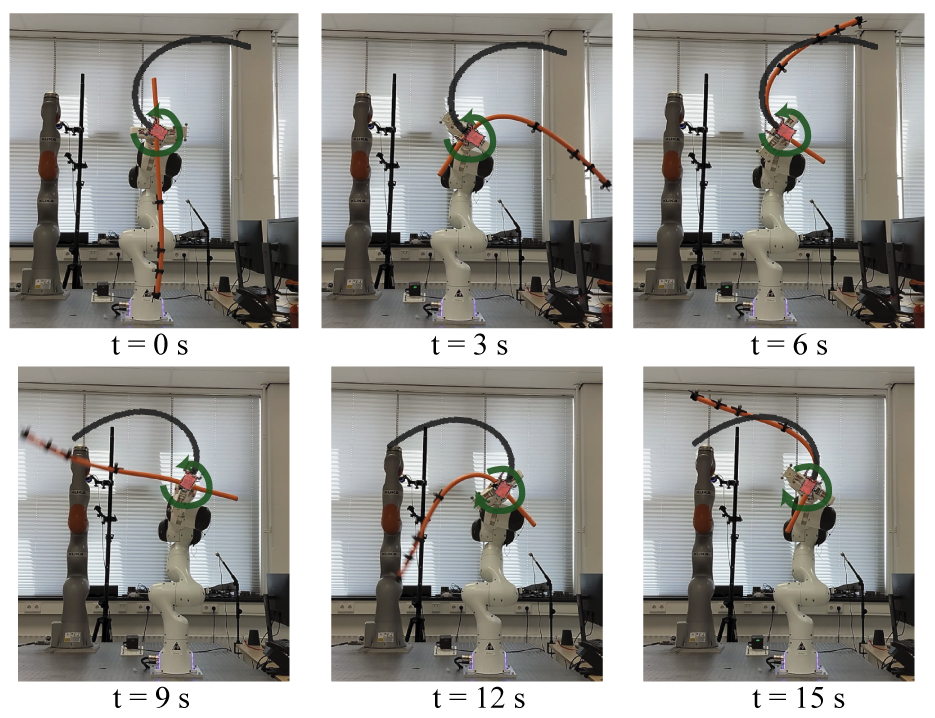}
  \caption{Franka–cable swing-up and transition between two equilibrium configurations to be regulated (black lines).}
  \label{fig:panda_swing_grid}
\end{figure}


\begin{figure}[h]
  \centering  \includegraphics[width=0.92\linewidth,keepaspectratio]{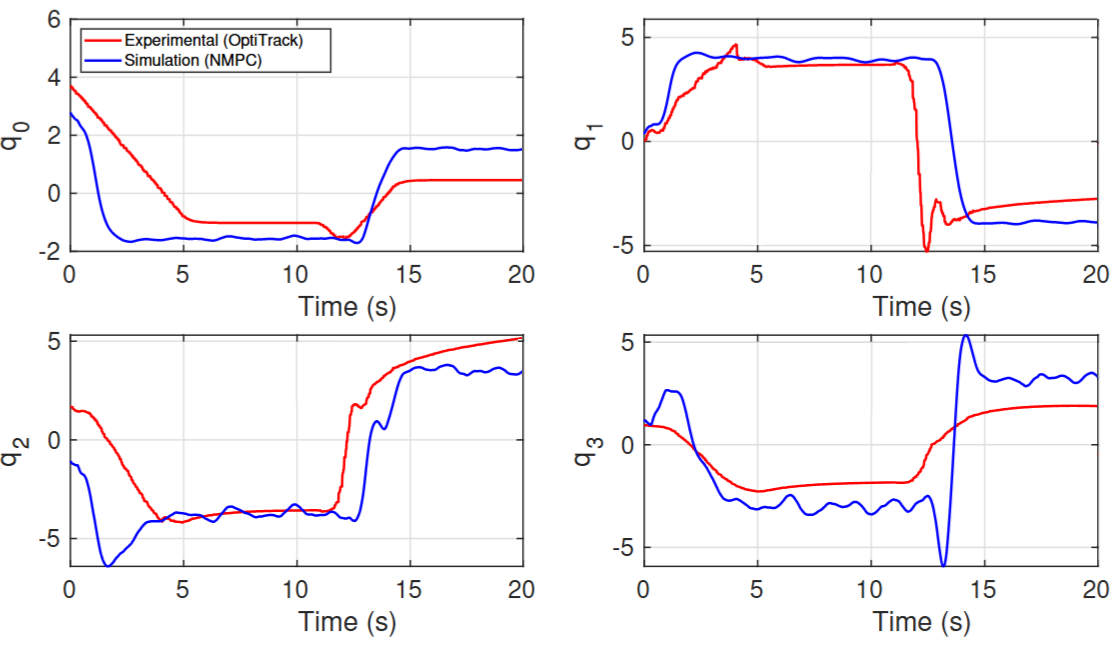} %
  \caption{State trajectories ($q_0$–$q_3$): comparison between the NMPC/DDP simulation and the experiment.}
  \label{fig:nmpc_vs_optitrack_states}
\end{figure}

\section{Conclusions}

This work demonstrates that combining Differential Dynamic Programming (DDP) with Nonlinear Model Predictive Control (NMPC) and differentiable Variable Strain modeling enables the stabilization of underactuated soft robots in low-stiffness regimes where traditional collocated and non-collocated controllers fail. By leveraging analytical gradients derived from RNEA-based dynamics and reduced-order modeling, the proposed framework achieves high computational efficiency while preserving physical model fidelity, offering a promising approach for stabilizing underactuated soft robots in such regimes.



\bibliography{mybibfile}
\bibliographystyle{ieeetr}

\end{document}